\documentclass[journal]{IEEEtran}
\usepackage{graphicx}
\usepackage{cite}
\usepackage{amsmath}
\usepackage{amsfonts}
\usepackage{amssymb}
\usepackage{subfigure}
\usepackage{psfrag}
\usepackage{arydshln}
\usepackage{xcolor}
\usepackage{amsthm,amssymb}
\usepackage{multirow}
\newtheorem{lemma}{Lemma}
\newcommand{\be}{\begin{equation}}
\newcommand{\ee}{\end{equation}}
\newcommand{\bea}{\begin{eqnarray}}
\newcommand{\eea}{\end{eqnarray}}
\newcommand{\bm}{\mathbf}

\newcommand{\by}{{\bm y}}

\newcommand{\bbf}{{\bm f}}

\newcommand{\bAn}{\mbox{\boldmath $\cal A$}}
\newcommand{\bLn}{\mbox{\boldmath $\cal L$}}
\newcommand{\bBn}{\mbox{\boldmath $\cal B$}}
\newcommand{\bCn}{\mbox{\boldmath $\cal C$}}

\newcommand{\bD}{{\bf D}}

\newcommand{\bS}{{\bf S}}

\newcommand{\bU}{{\bf U}}

\newcommand{\bH}{{\bm H}}

\newcommand{\bI}{{\bm I }}
\newcommand{\bJ}{{\bm J }}

\newcommand{\bLambda}{\mbox{\boldmath$\Lambda$}}

\newcommand{\bPsi}{\mbox{\boldmath{$\Psi$}}}

\begin{document}

\title{Spectral Domain Spline Graph Filter Bank}

\author{Amir Miraki, Hamid Saeedi-Sourck, Nicola Marchetti, and Arman Farhang
\thanks{
A. Miraki and H. Saeedi-Sourck are with EE Department, Yazd University, Yazd, Iran, 89158-18411 (e-mail: amir.miraki@stu.yazd.ac.ir, saeedi@yazd.ac.ir). N. Marchetti is with Trinity College Dublin, Ireland (e-mail: nicola.marchetti@tcd.ie)
A. Farhang is with Maynooth University, Maynooth, Ireland (e-mail: arman.farhang@mu.ie).}}         
\maketitle

\begin{abstract}
In this paper, we present a structure for two-channel spline graph filter bank with spectral sampling (SGFBSS) on arbitrary undirected graphs. Our proposed structure has many desirable properties; namely, perfect reconstruction, critical sampling in spectral domain, flexibility in choice of shape and cut-off frequency of the filters, and low complexity implementation of the synthesis section, thanks to our closed-form derivation of the synthesis filter and its sparse structure. These properties play a pivotal role in multi-scale transforms of graph signals. Additionally, this framework can use both normalized and non-normalized Laplacian of any undirected graph. We evaluate the performance of our proposed SGFBSS structure in nonlinear approximation and denoising applications through simulations. We also compare our method with the existing graph filter bank structures and show its superior performance.
\end{abstract}


\section{Introduction}
\IEEEPARstart{G}{raph} signal processing (GSP) extends classical signal processing to enable analysis of irregularly structured data on the vertices of an underlying graph, \cite{Shuman, GSP}. In recent years, GSP has been utilized in plethora of real-life applications such as data processing in social, transport, economic, biological and sensor networks  \cite{GSP}. High dimensional nature of data in these networks necessitates multirate signal analysis by construction of filter banks on graphs for different purposes such as denoising, compression, and  data classification, \cite{LSGFF}.

Graph filter bank (GFB) was first proposed for special types of graphs, namely, tree,
${\Omega}$-structure, circulant, and bipartite graphs \cite{tree,omega-struc,Ekambaram4,ortega1}. The authors in \cite{ortega1} proposed a two-channel critically-sampled GFB structure with quadrature mirror filters (GraphQMF) satisfying perfect reconstruction (PR) property for signals on bipartite graphs. This method is applicable to any arbitrary graph through a bipartite subgraph decomposition leading to a high computational complexity. Alternative filter design methods for the GFB structure in \cite{ortega1} with biorthogonal and frequency
conversion based filters were proposed in \cite{ortega2} and \cite{GraphFC}, known as GraphBior and GraphFC, respectively. An $M$-channel oversampled extension of \cite{ortega1} was presented in \cite{sakyama1}. In a more recent work, \cite{Pavez}, the results of \cite{ortega1} are extended to arbitrary graphs, without the need for bipartite subgraph decomposition, using a different definition of graph Fourier transform (GFT). The authors in \cite{subgraph} decomposed an arbitrary graph into several subgraphs. They applied local GFT to each subgraph and obtained a GFB with PR property (SubGFB). An  $M$-channel critically sampled GFB (CSFB) on arbitrary graphs was introduced in \cite{shuman1}, where the synthesis filters in each subband were replaced with interpolation operators. Authors in \cite{Anis} proposed a critical sampling method for two-channel filter bank on an arbitrary graph where the PR condition was only satisfied for bipartite graphs. Another GFB for arbitrary graphs is spline graph filter bank (SGFB) \cite{Ekambaram1}. The key difference between SGFB and other GFB structures is in substitution of synthesis filters with an inverse filter, which simplifies the GFB design, see Fig.~\ref{fig1}.

Unlike classical filter banks, down/upsampled signal in GFB has a considerably different spectrum from that of the original signal, except for bipartite graphs \cite{Tanaka1,Tanaka5}. This is a big challenge for multiscale analysis and processing on arbitrary graphs. To deal with this challenge, different approaches have emerged, \cite{jiang1,jiang2,Tanaka3,Tanaka4}. The authors in \cite{jiang1} and \cite{jiang2} proposed a GFB structure without down/upsampling that leads to a large computational load. In contrast, the  authors in \cite{Tanaka3} and \cite{Tanaka4}, take a more interesting approach and perform down/upsampling operations in spectral domain. This idea led to a critically-sampled GFB structure with spectral sampling (GraphSS) that is applicable to arbitrary graphs while satisfying the PR condition. The concept of spectral sampling, its superior performance to vertex domain methods, \cite{Tanaka1}, and the results of \cite{Tanaka3} and \cite{Tanaka4} are among the main motivations for extending SGFB, \cite{Ekambaram1}, from vertex domain to spectral domain in this paper. In \cite{Ekambaram1}, the designed analysis filters do not have desirable passband/stop-band characteristics. Hence, a filter design method, known as modified SGFB (MSGFB), was proposed in \cite{Amir1}. However, SGFB has a number of limitations; namely, spectral issues resulting from down/upsampling for arbitrary graphs, deteriorated performance compared to spectral sampling-based GFBs, and a high computational complexity due to the dense matrix inversion and multiplication operations at the synthesis section. Thus, the main goal of this paper is to address all these limitations.

In this paper, we present a two-channel critically-sampled SGFB structure with spectral sampling (SGFBSS), see Fig.~\ref{fig2}. Our proposed SGFBSS structure satisfies PR condition for arbitrary graphs without the requirement of any subgraph partitioning or decomposition. Since down/upsampling operations are performed in the spectral domain, they do not lead to any spectral issues. Our proposed structure can use both normalized and non-normalized graph Laplacian. We find a closed-form for the inverse filter at the synthesis side which is sparse and hence it leads to a low complexity implementation. We also discuss filter design methodology and show that our proposal provides a large amount of flexibility in the choice of filter parameters such as their shape and cut-off frequency. Our numerical results demonstrate the effectiveness of the proposed SGFBSS structure for applications such as nonlinear approximation and denoising on arbitrary undirected graphs. We compare our method with the existing GFB structures in the literature and show its superior performance.

\textit{Notations and preliminaries:} Boldface uppercase, boldface lowercase and normal letters represent matrices, vectors and scalar quantities, respectively. The superscript $(\cdot)^{\rm T}$ denotes transpose operation. A graph ${\mathcal G}=(\mathcal V,\mathcal E)$ is defined with a set of nodes $\mathcal V$, a set of edges $\mathcal E$ and an adjacency matrix ${\mathbf A}$ that describes the graph connectivity. ${\mathbf D}$ is the diagonal degree matrix whose diagonal elements are defined as the sum of the elements on the respective row of ${\mathbf A}$. In this paper, we consider undirected graphs without self-loops, i.e., the elements on the main diagonal of ${\mathbf A}$ are all zero. The graph Laplacian matrix is defined as $\mathbf L \equiv  \bD - \mathbf A$. Normalized adjacency and Laplacian matrices are defined as $\bAn\equiv \bD^{{\rm  - 1/2}} \mathbf A\bD^{{\rm  - 1/2}}$ and $\bLn \equiv  {\mathbf D}^{{\mathbf  - 1/2}} {\mathbf L\mathbf D}^{{\rm  - 1/2}}=\bI_N  - \bAn$, respectively, where $\bI_N$ is the $N\times N$ identity matrix. For connected graphs, $\bLn$ is a real-valued symmetric matrix. Thus, using eigenvalue decomposition, it can be written as $\bLn=\mathbf U \bLambda \mathbf U^{\rm T}$, where $\bLambda={\rm diag}\big([\lambda_0,\ldots,\lambda_{N-1}]^{\rm T}\big)$ is the diagonal eigenvalue matrix with the diagonal elements $0=\lambda_0<\lambda_1\leq \cdots \leq \lambda_{N-1}\leq 2$, $\mathbf U =\left[ {\boldsymbol{\rm u}_{\rm 0}, \ldots, \boldsymbol{\rm u}_{N - 1}} \right]$ is a unitary matrix that contains the orthonormal eigenvectors $\boldsymbol{\rm u}_{n}$ on its columns and ${\mathbf U\mathbf U}^{{\rm T}}=\bI_N$. Considering the signal $\mathbf f = \left[ {f(0), \cdots, f(N-1)}\right]^{\rm T}$ where the sample $f(n)$ appears on the $n$th node of the graph, GFT of this signal is defined as $\bar {\mathbf f}\triangleq{\mathbf U}^{\rm T}{ \mathbf f}$. Equivalently, inverse GFT of $\bar{\mathbf f}$ can be obtained as ${\mathbf f = \mathbf U\bar{\mathbf f}}$. Finally, the filtered signal is expressed as $\tilde {\bbf} = \bH \bbf$, where $\bH = \bU\,\overline{ \bH} \bU^{\rm T} $ and the filter kernel $\overline {\bH}$  is a diagonal matrix with the elements $\bar{H}(n)$ on its main diagonal, i.e., $\overline{\bH}={\rm diag}\big([\bar{H}(0),\ldots,\bar{H}(N-1)]^{\rm T}\big)$.

\section{two-channel SGFB with vertex  sampling}\label{Sec:III}
\begin{figure}[t]
\centering \includegraphics[scale=.5]{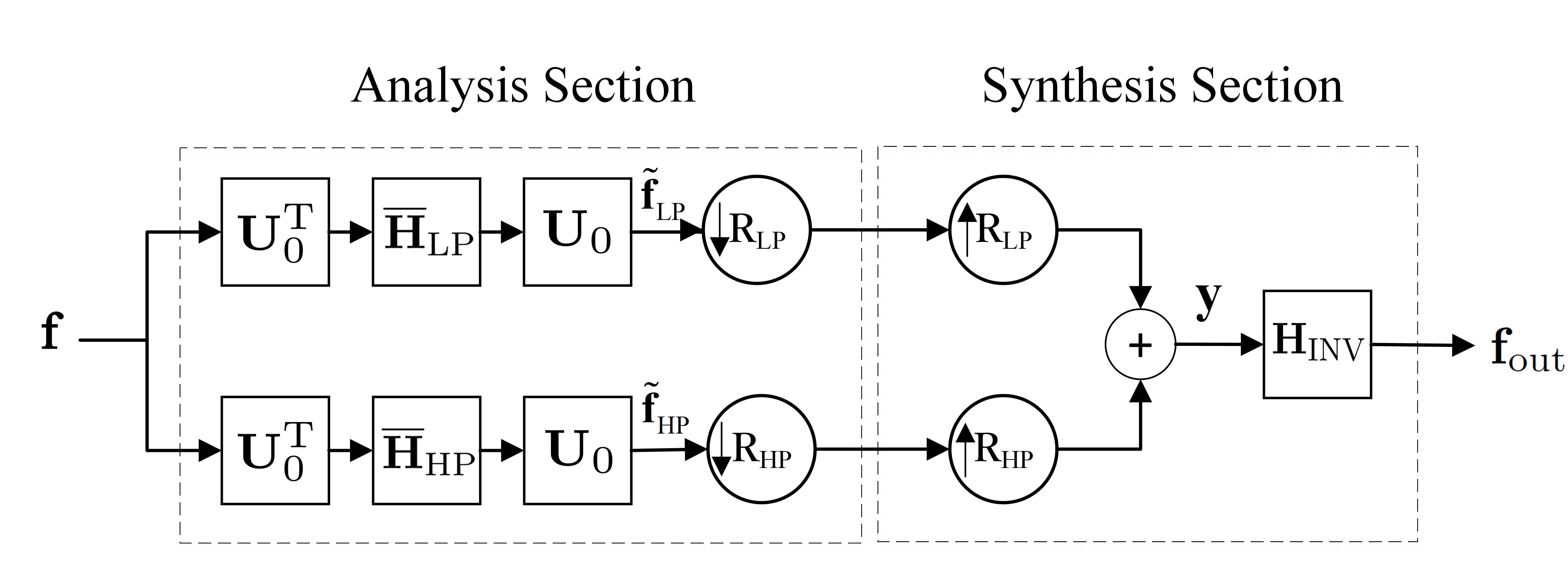}
  \caption{Two-channel SGFB with vertex  sampling \cite{Ekambaram1}.}
  \label{fig1}
\end{figure}

Figure~\ref{fig1} shows a two-channel SGFB with vertex domain sampling (SGFBVS) where the subscripts LP and HP refer to the low-pass and high-pass channels of the filter bank, respectively. Based on the results of \cite{Ekambaram1}, the filter $\bH_{\rm LP}$ can be obtained as a polynomial function of the normalized adjacency matrix with the order $J$. Thus, $\bH_{\rm LP}  =\frac{1}{2}(\bI_{\rm N}  + \bBn)$ where $\bBn=\sum_{ l  =  1}^J {w_l (\bAn )^l }$ and the weights $\{w_l\}_{l=1}^{J}$ are optimized to achieve a desired filter response. Also, $\bH_{\rm LP}$ can be diagonalized as $\bH_{\rm LP}= \bU_0\overline{\bH}_{\rm LP}\bU_0^{\rm T}$ where 
 \be 
 \label{filter-eq}
\overline{\bH}_{\rm LP}=\frac{1}{2}( \bI_N+\bPsi),
 \ee
$\bPsi=\sum_{l= 1}^J w_l(\mathbf {I}_N-\bLambda )^l={\rm diag}\{[\psi_0,\ldots,\psi_{N-1}]^{\rm T}\}$ with $\psi_n= \sum_{l{\rm   =  1}}^J {w_l (1-{\rm \lambda }_{n})^l }$ and ${\bar{ H}}_{\rm LP}(n)=\frac{1}{2}(1+\psi_n)$.

Additionally, $\bU_{0}$ and $\bLambda$ are the eigenvector and eigenvalue matrices of the normalized Laplacian matrix for the original graph, respectively. The high-pass filter $\bH_{\rm HP}$ can be constructed as $\bH_{\rm HP} =\bI_{\rm N}-\bH_{\rm LP}$ \cite{Ekambaram1}. After filtering, the signals $\tilde{\bbf}_{\rm LP}$ and $\tilde{\bbf}_{\rm HP}$ in low-pass and high-pass channels are downsampled by the factors $ {\rm R}_{\rm LP}$ and $ {\rm R}_{\rm HP}$, respectively. Therefore, the corresponding graph is reduced. To reduce the graph size, in this paper, we use the well-known Kron reduction method \cite{Kron}. In the synthesis section, the upsampled signals are combined and the signal $\by =\frac{1}{2} (\bI_{N}+\mathbf K\bBn)\bbf$ in the vertex domain is formed where $\mathbf {K}$ is a diagonal matrix with the diagonal elements $K(i,i) = 1$ if node $i$ is maintained after downsampling at LP channel, otherwise $K(i,i) = -1$. Finally, under the condition that $(\bI_{N}  + \mathbf K\bBn)$ is invertible, the original signal $\bbf$ is perfectly reconstructed as ${\bbf}_{\rm out} =\mathbf {H}_{\rm INV}\by$ where $\mathbf {H}_{\rm INV} =2(\bI_{N}  + \mathbf K\bBn)^{-1}$ \cite{Ekambaram1}. Hence, the weights, $w_l$, need to be designed to guarantee the invertibility of $(\bI_{N}  + \mathbf K\bBn)$.

\section{Two-channel SGFB with spectral  sampling}\label{Sec:IV}
\begin{figure}[t]
\centering \includegraphics[scale=.50]{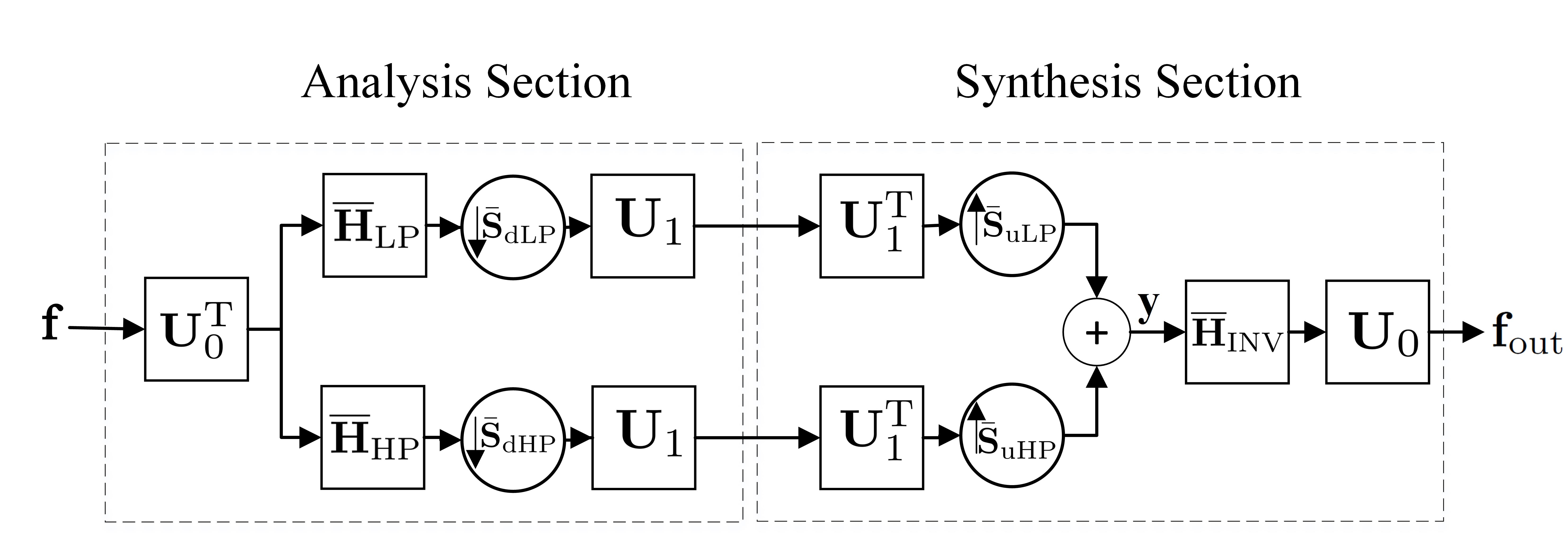}
  \caption{Proposed two-channel SGFB with spectral sampling.}
  \label{fig2}
\end{figure}
In this section, we propose an architecture for the two-channel SGFB based on the spectral sampling concept which was first introduced in \cite{Tanaka1}, see Fig.~\ref{fig2}. We derive PR conditions, present a filter design method and derive a closed-form for the inverse synthesis filter leading to a low complexity implementation.

Let us consider the downsampling matrices as
\begin{equation}
\begin{array}{l}
 {\bar{ \bS}}_{\rm dLP}  = \left[\bI_{N/2}\quad  \bJ_{N/2}  \right],
 ~{\bar{ \bS}}_{\rm dHP} = \left[ {\bI_{N/2} \quad  -\bJ_{N/2} } \right],
 \end{array}\label{sampl_matrix}
\end{equation}
where $\bJ_{N/2}$ is the counter identity matrix of size $N/2$, \cite{Tanaka1,Tanaka3} and upsampling matrices as ${\bar{ \bS}}_{\rm uLP}  = {\bar{ \bS}}_{\rm dLP}^{\rm T}$ and ${\bar{ \bS}}_{\rm uHP}  = {\bar{ \bS}}_{\rm dHP}^{\rm T}$.  Hence, from Fig.~\ref{fig2}, the spectral domain signal $\by$ can be obtained as
\be
\mathbf y\!=\!\left({\bar \bS}_{\rm uLP}\bU_1^{\rm{T}}\bU_1{\bar \bS}_{\rm dLP}{\overline\bH}_{\rm{ LP}}+{\bar \bS}_{\rm uHP}\bU_1^{\rm{T}}\bU_1{\bar \bS}_{\rm dHP}{\overline\bH}_{\rm{HP}}\right)\!\bU_0^{\rm{T}}\mathbf{f},\label{recons_eq}
\ee
where $\bU_0$ and $\bU_1$ are the unitary eigenvector matrices corresponding to the original and the reduced-size graphs, respectively \cite{Tanaka1}. Since, $\bU_1^{\rm T}\bU_1=\bI_{N/2}$, (\ref{recons_eq}) reduces to
\be\label{eq_3}
\mathbf y=\left({\bar \bS}_{\rm uLP}{\bar \bS}_{\rm dLP}{\overline\bH}_{\rm{ LP}}+{\bar \bS}_{\rm uHP}{\bar \bS}_{\rm dHP}{\overline\bH}_{\rm{HP}}\right)\!\bU_0^{\rm{T}}\mathbf{f}.
\ee
By substituting (\ref{filter-eq}) and (\ref{sampl_matrix}) into (\ref{eq_3}), we have
 \begin{align}\label{eqn:y_SS}
\mathbf y =\bCn\bU_0^{\rm{T}}\mathbf{f},
\end{align}
where the square matrix $\bCn=\bI_N +\bJ_N\bPsi$ is non-zero only on its main and anti-diagonal elements, i. e.,
\be \label{eqn:bCn}
\bCn =
 \begin{bmatrix}
  1     & 0 & \cdots & 0 &  \psi_{N-1}   \\
  0      & 1 & \cdots & \psi_{N-2} & 0  \\
  \vdots & \vdots &  \ddots  & \vdots & \vdots   \\
  0      & \psi_1  & \cdots  &1 & 0 \\
   \psi_0 & 0   & \cdots       &0 & 1
 \end{bmatrix}.
\ee

From Fig.~\ref{fig2} and using (\ref{eqn:y_SS}), the output signal of the synthesis section can be represented as
\be\label{eq-esbat1}
\bbf_{\rm out}=\bU_0{\overline \bH}_{\rm INV}\by=\bU_0{\overline \bH}_{\rm INV}\bCn\bU_0^{\rm{T}}\mathbf{f}.
\ee
From this equation, one may realize that opposed to the inverse filter in Fig.~\ref{fig1} that works in the vertex domain, the inverse filter in our proposed architecture operates in spectral domain. Based on (\ref{eq-esbat1}), the original signal $\bbf$ can be perfectly reconstructed when ${\overline \bH}_{\rm INV}=\bCn^{-1}$. Hence, the two-channel SGFBSS has PR property under the condition that $\bCn$ is invertible. The  square matrix $\bCn$ is invertible if and only if its rows  are linearly independent.

\begin{lemma} \label{Lemma:1}
The  square matrix $\bCn$ is invertible and the original signal $\bbf$ is perfectly reconstructed using (\ref{eq-esbat1}), if and only if
${\psi_n}\ne\frac{1}{\psi_{N-n-1}}, \forall~n\in\{0,\cdots,{N}/{2}-1\}$.
\end{lemma}

\begin{proof}
The special structure of the matrix $\bCn$ that is shown in (\ref{eqn:bCn}), suggests that this matrix always has $\frac{N}{2}$ independent rows. This is because $\bCn=[{\bf c }_0^{\rm T},\ldots,{\bf c }_{N-1}^{\rm T}]^{\rm T}$ is always comprised of $\frac{N}{2}$ pairs of symmetrical rows, ${\bf c }_n$ and ${\bf c }_{N-n-1}$ with non-zero entries on the same columns. Hence, this matrix has $N$ independent rows if and only if  each pair of symmetrical rows are linearly independent, i.e.,
\be\label{Ind1}
\nexists\;\alpha,\quad {\bf c }_{n}=\alpha{\bf c }_{N-n-1},\quad n=0,\cdots,N/2-1,
\ee
where $\alpha$ is a scalar \cite{Matrix_analysis}.
Let us assume there exists an $\alpha$ so that ${\bf c }_{n}=\alpha{\bf c }_{N-n-1}$. From (\ref{eqn:bCn}) and using (\ref{Ind1}), one may realize that $\psi_{n}=\frac{1}{\alpha}$ and $\psi_{N - n -1}=\alpha$. Consequently, the condition in (\ref{Ind1}) is satisfied by ${\psi _n}\ne\frac{1}{\psi _{N-n-1}}$.
\end{proof}

\subsection{Filter design methodology}
The shape of the filters plays a crucial role for signal decomposition into different spectral bands. In this section, a pair of analysis filters with desirable frequency responses is proposed for SGFBSS. We consider the spectral kernel as
\begin{equation}
 {\bar H}_{\rm LP}(n)=\begin{cases}
    1, & \text{if $\lambda _n \le \lambda_{\rm cut}$},\\
    \varepsilon, & \text{if $\lambda _n  > \lambda_{\rm cut}$}.
\end{cases}
\end{equation}
For $\lambda_{\rm cut}  = \lambda _{\frac{N}{2}-1}$ and $\varepsilon=0$, we have the exact ideal low-pass filter. Using (\ref{filter-eq}), we have ${\psi _n}=2{\bar H}_{\rm LP} (n)-1, n=0,\cdots,N-1$. Hence, for exact ideal filter, $\psi_0=\ldots=\psi_{\frac{N}{2}-1}=1$ and $\psi_{\frac{N}{2}}=\ldots=\psi_{N-1}=-1$. For a higher flexibility, cut-off frequency can be variable within the range ${\lambda _0  < \lambda_{\rm cut}  \le \lambda _{\frac{N}{2}-1} }$ and then $\varepsilon\neq 0$ to satisfy PR condition, as mentioned earlier.

Non-ideal filters are sometimes preferred when the eigenvalue distribution of the variation operator is irregular \cite{Tanaka3}. In this paper, we assume a Butterworth filter with order ${\beta}$ as ${\bar H}_{\rm LP}(n) = (1+({\lambda_n}/{\lambda_{\rm cut} })^{2{\beta}} )^{-0.5}$. The cut-off frequency does not have any limitations in our proposed SGFBSS structure. Thus, $\lambda_{\rm cut}$ and ${\beta}$ are the design parameters.
The selection of cut-off frequency for analysing filters in GFBs has received less attention in the literature,\cite{ortega1, ortega2, GraphFC, Tanaka3}. In GFBs for bipartite graphs, $\lambda_{\rm cut}=\lambda_{\frac{N}{2}-1}$ is assumed. This is reasonable due to the order of eigenvalues for bipartite graphs. However, the shapes of the filters for arbitrary graphs are important and require more flexibility for cut-off frequency which is satisfied by using our method.

\subsection{Low-complexity implementation}\label{LCI}
As it was mentioned in the proof of Lemma \ref{Lemma:1}, $\bCn$ contains $\frac{N}{2}$ pairs of symmetrical rows, ${\bf c}_n$ and ${\bf c }_{N-n-1}$ with non-zero elements only on two similar columns. As a result, the linear system of equations in (\ref{eqn:y_SS}) that is defined by the coefficient matrix $\bCn$, can be broken into $\frac{N}{2}$ isolated linear systems of equations with only two unknowns in each. Consequently, $\bCn^{-1}$ can be easily obtained by inverting $\frac{N}{2}$ matrices of size $2\times 2$ each. Hence, $\bCn^{-1}={\widetilde \bPsi}(\bI_N -\bJ_N\bPsi)$ where ${\widetilde \bPsi}={\rm diag}\{[\widetilde\Psi(0),\ldots,\widetilde\Psi(N-1)\}$ is a diagonal matrix with diagonal elements $\widetilde\Psi(n)=\widetilde\Psi(N-n-1)=1/(1-\psi_n\psi_{N-n-1})$ for $n=0,\ldots,N/2-1$ that are reciprocals of the determinants of the corresponding $2\times 2$ matrices. This simple closed-form for $\bCn^{-1}$, significantly reduces the computational complexity of the matrix inversion especially for large graphs.

To compare the complexity of our proposed SGFBSS with other existing solutions  in \cite{shuman1,Ekambaram1,Tanaka3,Amir1}, we focus on filtering and sampling. Both the spectral approaches of SGFBSS and the method in \cite{Tanaka3} have the same complexity for filtering in the analysis section and sampling. Interestingly, both GFBs have the same complexity in the  synthesis section. In particular, spectral domain filtering in \cite{Tanaka3} requires $2\mathcal{O}(N)$ number of multiplications. To pass the signal $\by$ through the inverse filter ${\overline \bH}_{\rm INV}$, our proposed SGFBSS method requires $2\mathcal{O}(N)$ rather than $\mathcal{O}(N^2)$ multiplications as compared with the SGFBVS method in \cite{Ekambaram1} and \cite{Amir1}. However, the class of vertex sampling methods such as the ones in \cite{shuman1,Ekambaram1} and \cite{Amir1} require lower complexity than methods with spectral sampling such as the one proposed in this paper and \cite{Tanaka3}. This is the cost to pay for the better performance of the spectral sampling based methods.

\begin{figure}
  \centering
  \subfigure[Sensor graph ($N$=100)]{\includegraphics[scale=0.14]{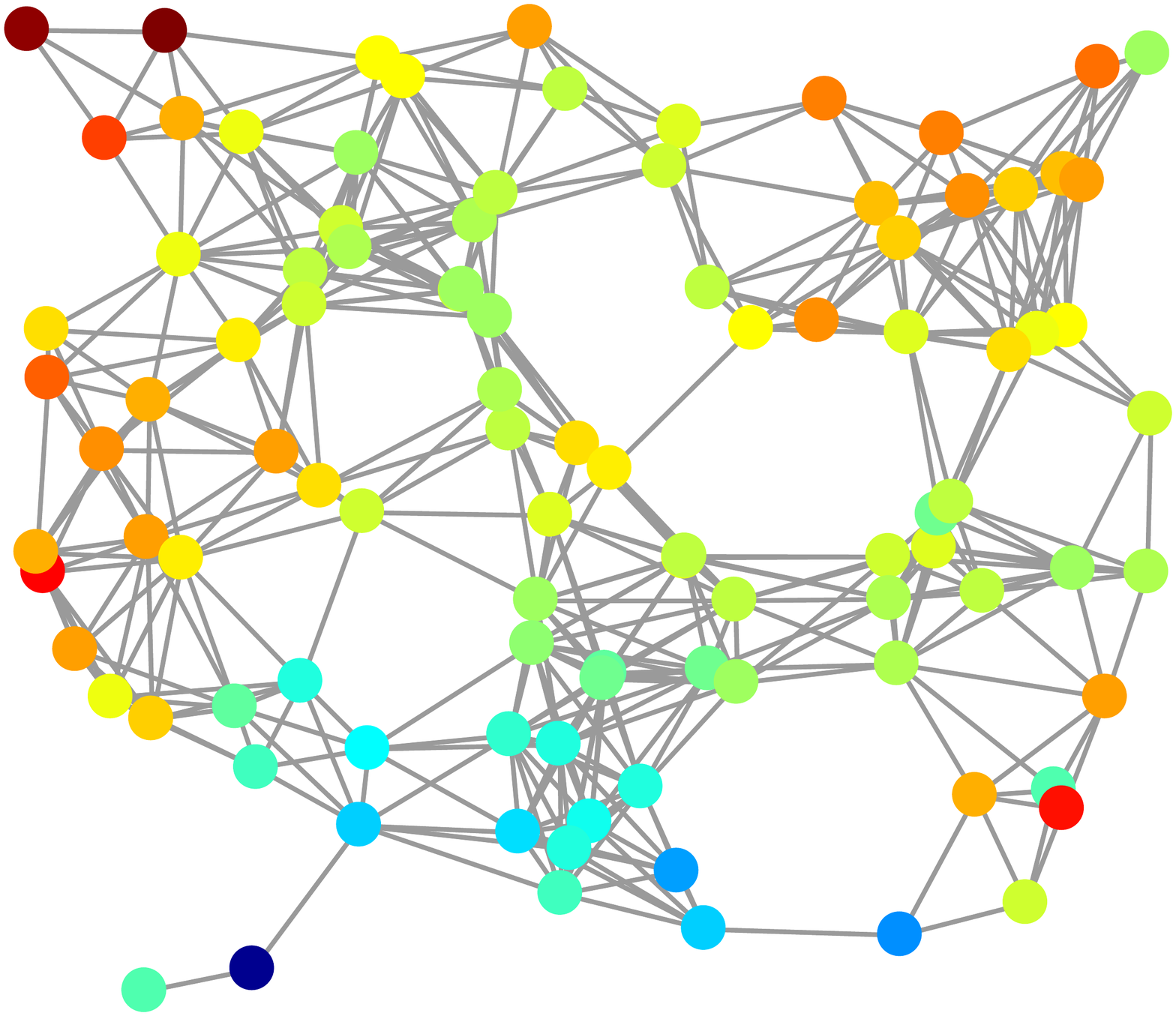}\label{fig2:sensor}}
\subfigure[Community graph ($N$=400)]{\includegraphics[scale=0.14]{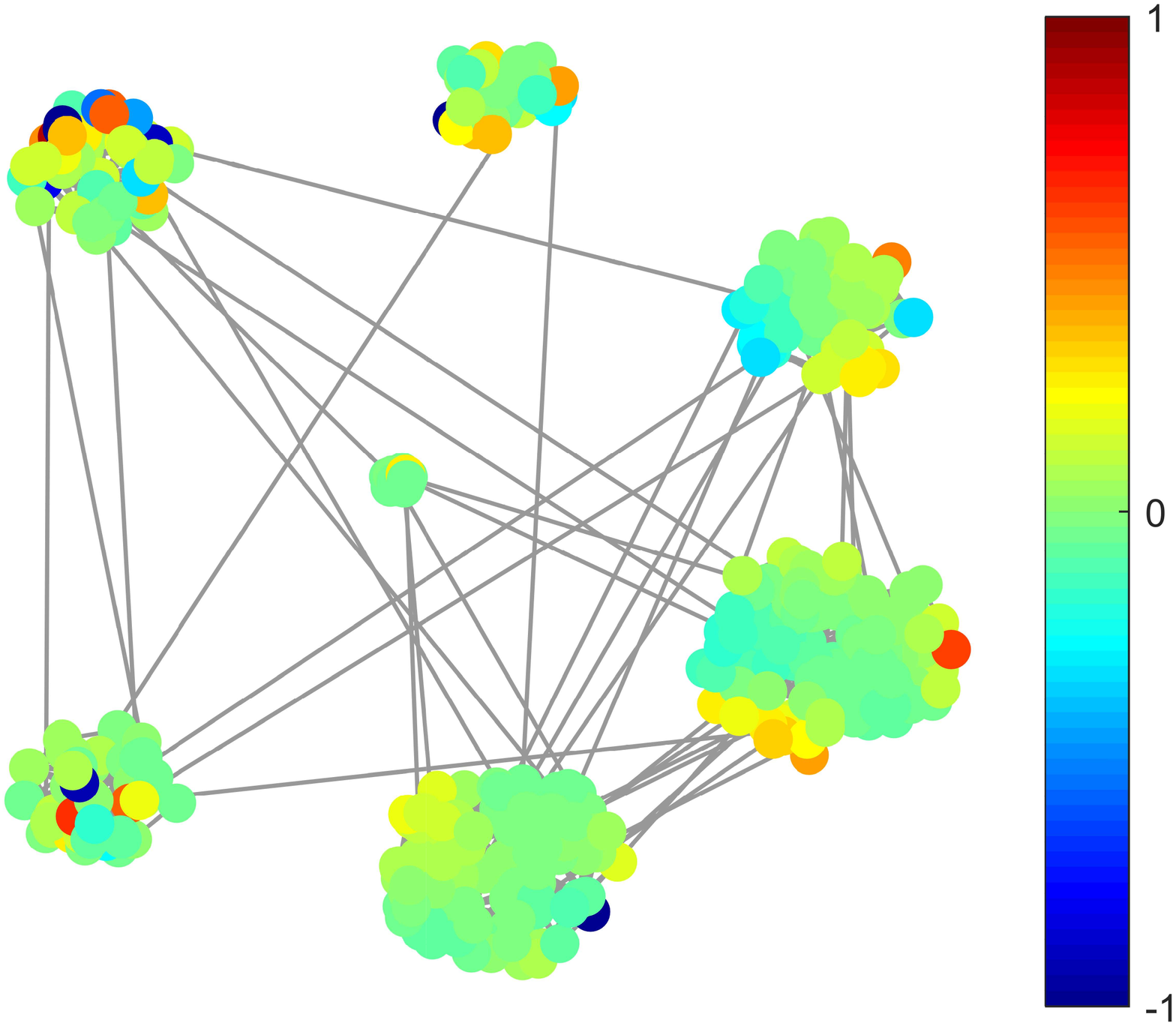}\label{fig2:community}}
  \caption{Graph signal in vertex domain.}
  \label{fig3}
\end{figure}
\begin{figure}
  \centering
   \hspace{-0.7cm} \subfigure[Sensor graph signal]{\includegraphics[scale=0.17]{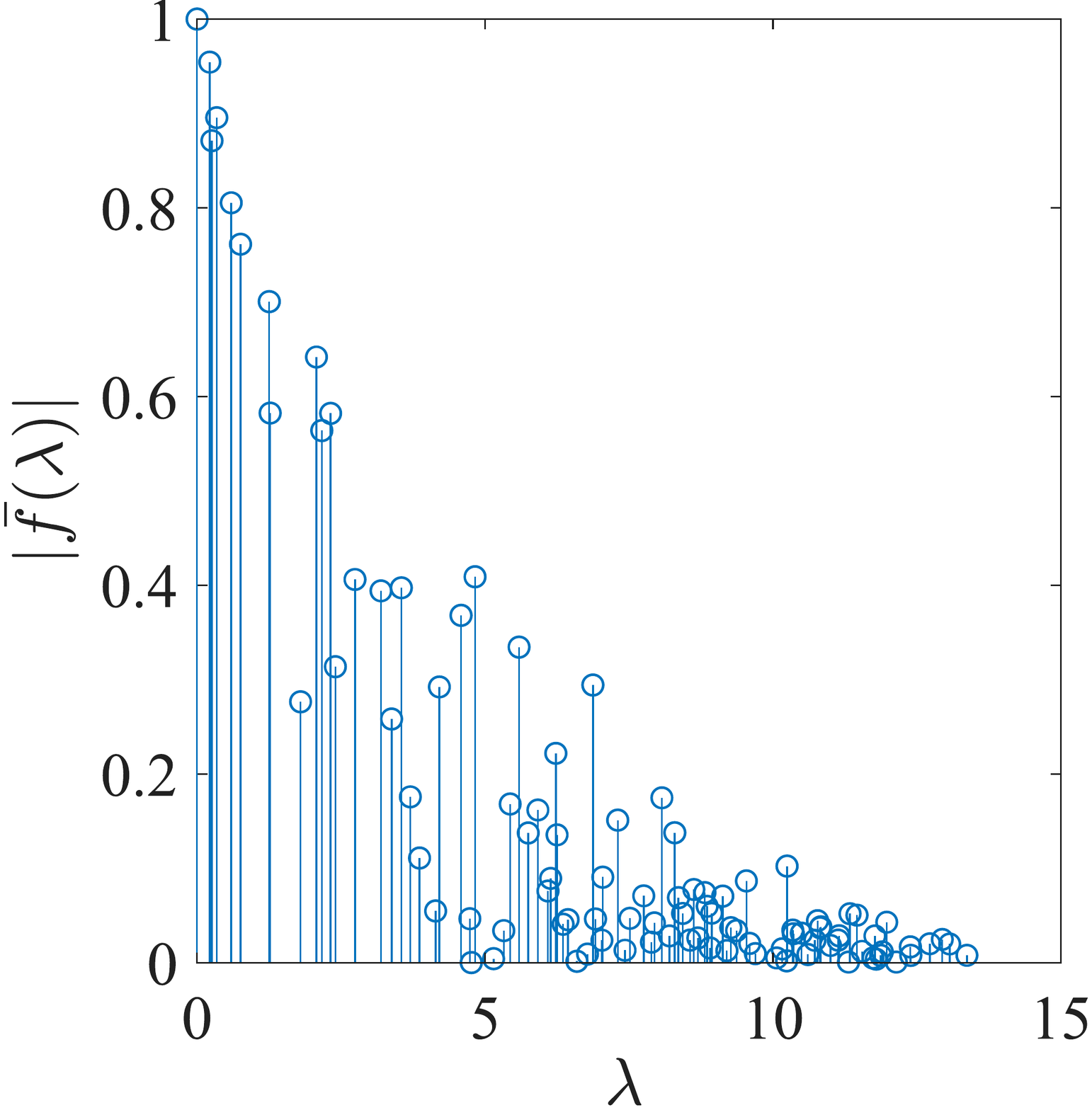}\label{fig4:sensor}}
    \hspace{-0.7cm}\subfigure[Community graph signal]{\includegraphics[scale=0.17]{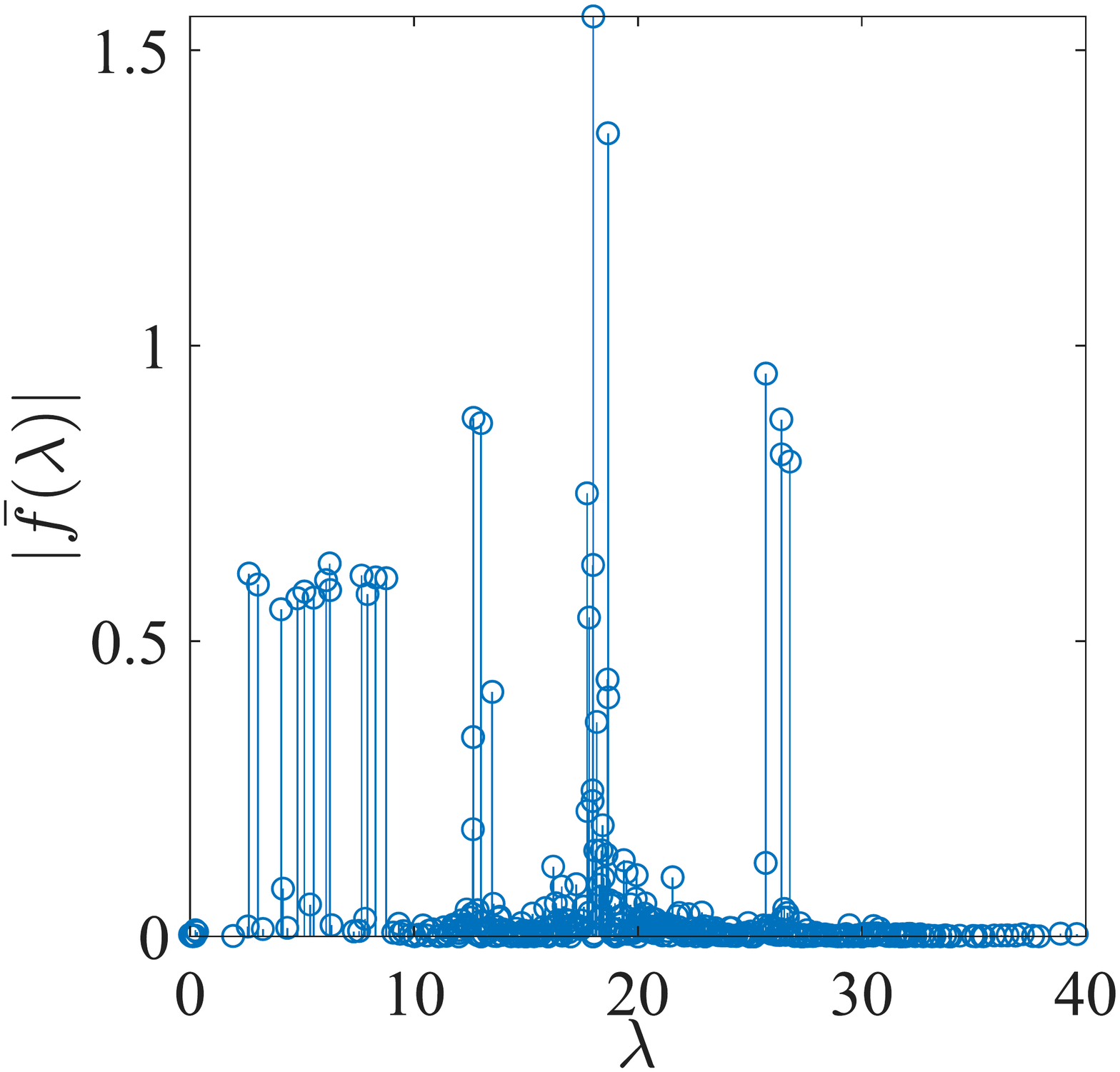}\label{fig4:community}}
  \caption{Graph signal in spectral domain.}
  \label{fig4}
\end{figure}
\section{Simulation Results}\label{Sec:Simulation}
In this section, we evaluate and compare the performance of our proposed SGFBSS structure with the existing GFBs in the literature, \cite{ortega1, Pavez, subgraph, shuman1, Tanaka3, Amir1, ortega2, GraphFC}. We have used GSPbox in MATLAB, \cite{GSPbox}, for graph generation, GSP operations, and visualizations. Similar to \cite{Tanaka3}, we consider two different graph signals with vertex and spectral representations that are shown in Figs.~\ref{fig3}~and~\ref{fig4}, respectively.
Figs.~\ref{fig4:sensor}~and~\ref{fig4:community} illustrate an approximately smooth signal and a localized signal in the spectral domain on the sensor and community graphs, respectively. As mentioned before, we have  freedom in choosing the analysis filters and cut-off frequencies. We assume ideal and Butterworth filters depicted in  Fig.~\ref{fig5} for the community graph as an example. Fig.~\ref{fig5:ideal} and  Fig.~\ref{fig5:butt} show the ideal filters with different cut-off frequencies and Butterworth filters with different orders for $\lambda_{\rm cut}=\lambda _{\frac{N}{2}-1}$, respectively. In the following, the performance of our proposed method in nonlinear approximation and denoising is evaluated.
\begin{figure}[t]
  \centering
  \subfigure[SGFBSS-I]{\includegraphics[scale=0.15]{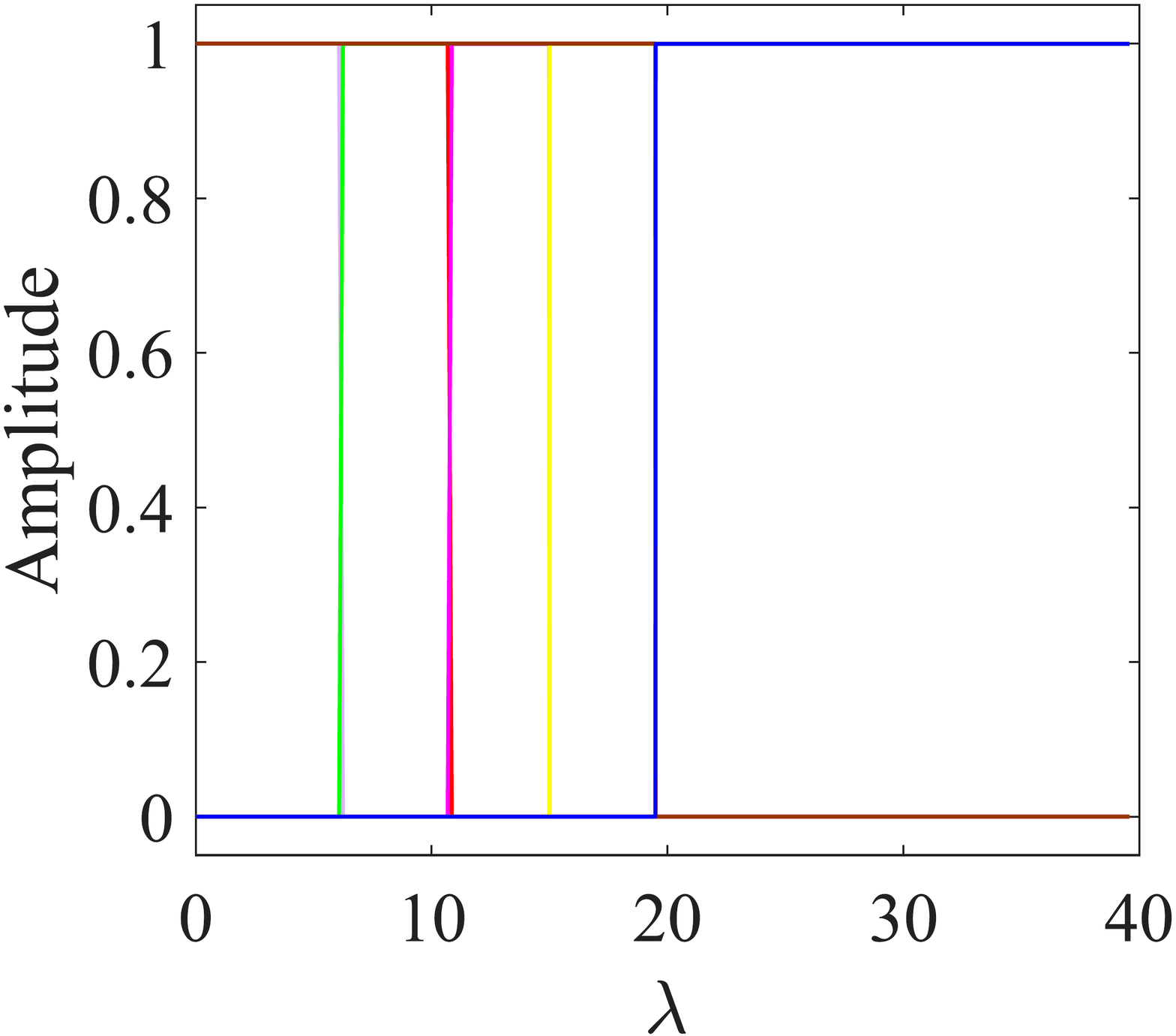}\label{fig5:ideal}}
\subfigure[SGFBSS-B]{\includegraphics[scale=0.15]{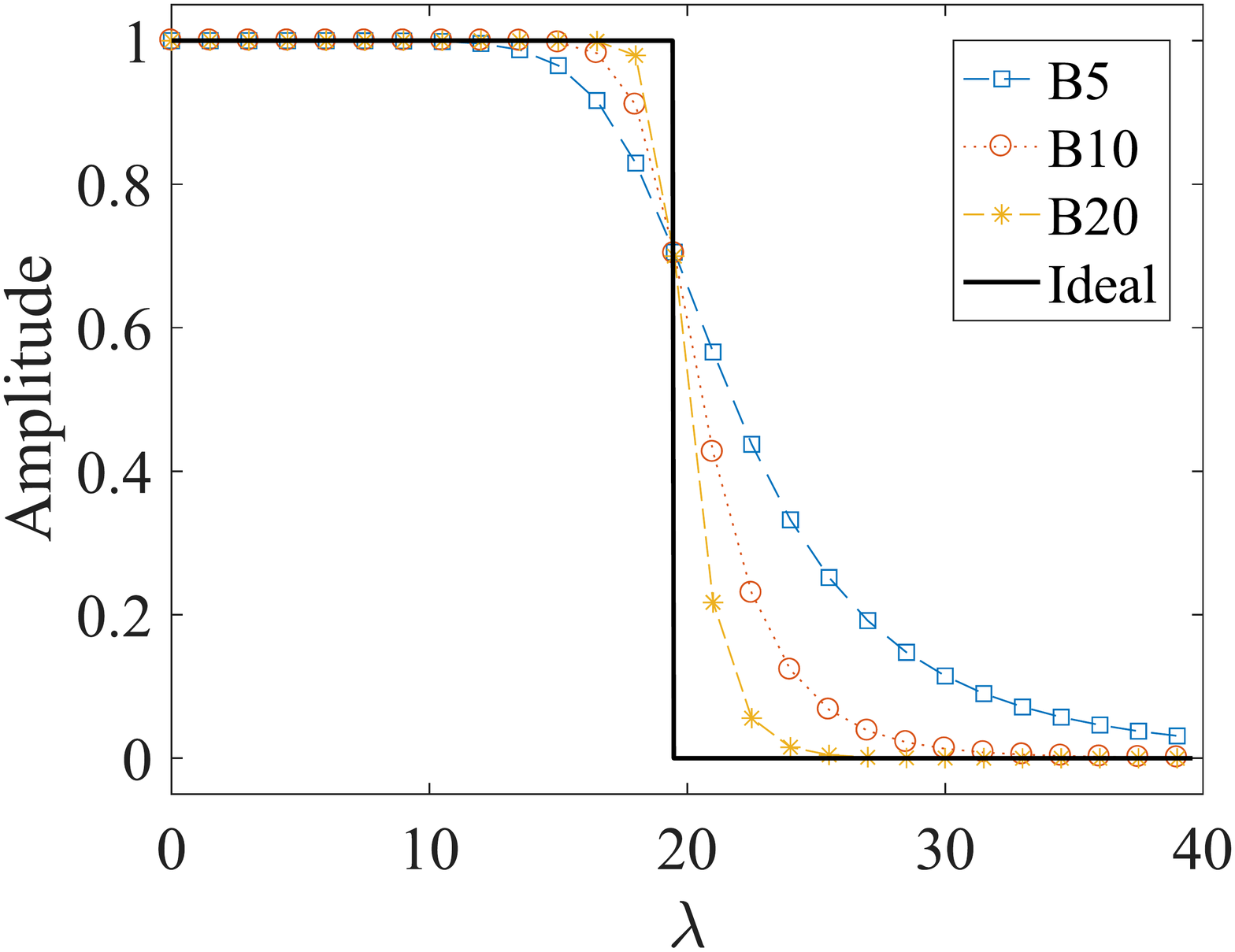}\label{fig5:butt}}
  \caption{Filter sets for SGFBSS (a) Ideal filters, (b)  Butterworth filters.}
  \label{fig5}
\end{figure}
\subsection{Nonlinear Approximation}
In nonlinear approximation, a fraction of the coefficients with high absolute values are kept and the rest are set to zero. In Fig.~\ref{fig6}, we compare our proposed SGFBSS structure with GraphQMF \cite{ortega1}, GraphBior \cite{ortega2}, GraphFC \cite{GraphFC}, MQGFB \cite{Pavez}, SubGFB \cite{subgraph}, CSFB \cite{shuman1}, GraphSS \cite{Tanaka3} and MSGFB \cite{Amir1}. This figure shows the resulting signal to noise ratio (SNR) versus the fraction of remaining coefficients for both sensor and community graphs. SGFBSS-I and SGFBSS-B20 represent SGFBSS with ideal and order 20 Butterworth filters, respectively. Our results in Fig.~\ref{fig6} show that spectral sampling based methods outperform the ones with vertex sampling that are shown with solid and dashed lines, respectively. In particular, our proposed SGFBSS structure achieves a significantly improved performance compared to its counterpart SGFB with vertex sampling, \cite{Amir1}. Furthermore, while having a superior performance to all the existing methods, our proposed structure leads to about the same performance as GraphSS, \cite{Tanaka3}.

\begin{figure}[t]
  \centering
  \includegraphics[scale=0.33]{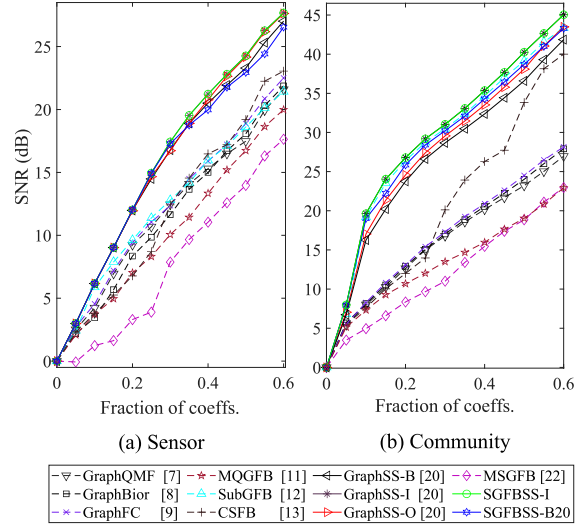}
  \caption{Results of nonlinear approximation.}
  \label{fig6}
\end{figure}

\begin{table}[t]
\caption{Denoising results. Average of 1000 runs.}
    \centering
    \begin{tabular}{l|c|c|c|c|c}         &\!\!\! Methods\!\!\!&\!\!\! \!\!$\sigma=1/8$ \!\!\! & \! \!\!\!\!$\sigma=1/4$ \!\!\! & \!\!\! \!\!$\sigma=1/2$ \!\!\! & \!\!\!\!\!\!\!  $\sigma=1$ \!\!\! \\ \hline\hline
        Sensor &\begin{tabular}{l}
              \!\!\!\!GraphSS-I \cite{Tanaka3}\!\!\!  \\
              \!\!\!\!GraphSS-O \cite{Tanaka3}\!\!\!  \\
              \!\!\!\!GraphSS-B \cite{Tanaka3}\!\!\!  \\
              \!\!\!\!CSFB \cite{shuman1}\!\!\! \\
              \!\!\!\!MQGFB \cite{Pavez}\!\!\! \\ \hdashline
              \!\!\!\!SGFBSS-I\!\!\!\\
              \!\!\!\!SGFBSS-B5\!\!\!\\
              \!\!\!\!SGFBSS-B10\!\!\!\\
              \!\!\!\!SGFBSS-B20\!\!\!\\
        \end{tabular} &\begin{tabular}{l}
             \!\!\!0.41  \\
            \!\!\!0.37  \\
            \!\!\!0.24  \\
            \!\!\!0.41  \\
           \!\!\!-2.90  \\\hdashline
            \!\!\!0.41  \\
            \!\!\!0.51  \\
            \!\!\!{\bf 0.52}  \\
            \!\!\!0.47  \\
        \end{tabular}&\begin{tabular}{l}
             \!\!\!1.28  \\
            \!\!\!1.28  \\
            \!\!\!1.20  \\
            \!\!\!1.28 \\
            \!\!\!-0.05 \\\hdashline
            \!\!\!1.28  \\
            \!\!\!{\bf 1.36}  \\
            \!\!\!1.31  \\
            \!\!\!1.33  \\
        \end{tabular}&\begin{tabular}{l}
             \!\!\!4.89  \\
            \!\!\!4.89  \\
            \!\!\!4.84  \\
            \!\!\!4.89 \\
            \!\!\!4.76 \\\hdashline
            \!\!\!4.89  \\
            \!\!\!{\bf 5.10}  \\
            \!\!\!4.96  \\
            \!\!\!4.94  \\
        \end{tabular}&\begin{tabular}{l}
            \!\!\!10.34  \\
            \!\!\!10.40  \\
            \!\!\!10.37  \\
            \!\!\!10.34 \\
            \!\!\!10.20 \\\hdashline
            \!\!\!10.34  \\
            \!\!\!{\bf 10.82}  \\
            \!\!\!10.69  \\
            \!\!\!10.55  \\
        \end{tabular}  \\
        \hline\hline
         \!\!\!\!Community  &\begin{tabular}{l}
              \!\!\!\!GraphSS-I \cite{Tanaka3}\!\!\!  \\
              \!\!\!\!GraphSS-O \cite{Tanaka3}\!\!\! \\
              \!\!\!\!GraphSS-B \cite{Tanaka3}\!\!\!  \\
             \!\!\!\!CSFB\cite{shuman1}\!\!\!  \\
             \!\!\!MQGFB \cite{Pavez}\!\!\! \\\hdashline
              \!\!\!\!SGFBSS-I\!\!\!\\
              \!\!\!\!SGFBSS-B5\!\!\!\\
              \!\!\!\!SGFBSS-B10\!\!\!\\
              \!\!\!\!SGFBSS-B20\!\!\!\\
        \end{tabular}&\begin{tabular}{l}
             \!\!\!{\bf 6.04}  \\
            \!\!\!5.39  \\
            \!\!\!5.07  \\
            \!\!\!{\bf 6.04} \\
            \!\!\!0.64 \\\hdashline
            \!\!\!{\bf 6.04}  \\
            \!\!\!3.78  \\
            \!\!\!5.28  \\
            \!\!\!5.94  \\
        \end{tabular}&\begin{tabular}{l}
             \!\!\!4.85  \\
            \!\!\!4.50  \\
            \!\!\!4.40  \\
            \!\!\!4.85 \\
            \!\!\!3.65 \\\hdashline
            \!\!\!4.85  \\
            \!\!\!4.06  \\
            \!\!\!4.72  \\
            \!\!\!{\bf 4.86}  \\
        \end{tabular}&\begin{tabular}{l}
            \!\!\!7.68  \\
            \!\!\!7.53  \\
            \!\!\!7.49  \\
            \!\!\!7.68  \\
            \!\!\!7.12 \\\hdashline
            \!\!\!7.68  \\
            \!\!\!{\bf 7.80}  \\
            \!\!\!7.76  \\
            \!\!\!7.79  \\
        \end{tabular}&\begin{tabular}{l}
             \!\!\!12.11  \\
            \!\!\!12.13  \\
            \!\!\!11.86  \\
            \!\!\!12.11 \\
            \!\!\!11.71 \\\hdashline
            \!\!\!12.11  \\
            \!\!\!{\bf 13.20}  \\
            \!\!\!12.72  \\
            \!\!\!12.38
        \end{tabular}
       \end{tabular}
    \label{tab:my_label}
\end{table}

\subsection{Denoising}
In this subsection, we evaluate the performance of our proposed SGFBSS structure for noise suppression. We consider the noisy signal $\bbf_{\rm noisy}=\bbf+{\boldsymbol\xi}$ where ${\boldsymbol\xi}$ is the zero-mean white Gaussian noise vector with the standard deviation $\sigma$. We compare both SGFBSS-I and SGFBSS-B (with different orders $5$, $10$, and $20$) with MQGFB utilizing ``lazy" bi-orthogonal filters \cite{Pavez}, CSFB \cite{shuman1} and GraphSS \cite{Tanaka3}. We use non-normalized Laplacian matrix as the variation operator. Cut-off frequency is chosen as $\lambda_{\rm cut}=\lambda _{\frac{N}{2}-1}$ and the coefficients at low and high frequency channels are hard-thresholded with $T = 3\sigma$. Denoising results for the graph signals of Fig.~\ref{fig3} are shown in Tab.~\ref{tab:my_label} in terms of signal-to-noise ratio (SNR) improvement in dB, i.e.,  $ \Delta _{{\rm SNR}}  = 10\log _{10} \Big(\frac{{\left\|{\boldsymbol \xi}  \right\|_2^2 }}{{\left\| {{\bold{ f}}_{{\rm denoised}} {\rm  - \bold{ f}}} \right\|_2^2 }}\Big)$ \cite{LSGFF}. In this table, different noise levels are considered and the largest $\Delta _{{\rm SNR}} $ values are represented in bold. Our results show that in many cases, non-ideal filters outperform the ideal ones. For the sensor graph, our proposed method is superior to all the other methods for different noise levels. Similar results are achieved for the community graph while only for the lowest noise level, $\sigma=1/8$, the methods in \cite{Tanaka3} and \cite{shuman1} lead to the same performance as our proposed method SGFBSS-I. Finally, for all the other noise levels, our proposed SGFBSS-B provides the highest SNR improvement.

\section{Conclusion}\label{Sec:Conclusion}
In this paper, we introduced a two-channel critically-sampled SGFB structure based on the spectral sampling concept. This structure is applicable to any arbitrary undirected graph and can use both normalized and non-normalized graph Laplacian. We derived PR condition and discussed filter design aspects. Our proposed structure provides a large amount of flexibility in terms of shape and cut-off frequency of the filters. We derived a closed-form for the synthesis filter that led to a low complexity implementation of the synthesis section. Our simulation results demonstrate the superior nonlinear approximation and noise suppression performance of our proposed method compared to the existing GFB structures in the literature.

\IEEEpeerreviewmaketitle
\ifCLASSOPTIONcaptionsoff
  \newpage
\fi

\end{document}